\documentclass[11pt, letterpaper]{article}
\usepackage{fullpage}
\usepackage{amsthm}
\usepackage{amsmath,amssymb,amsfonts,nicefrac}
\usepackage{xspace}
\usepackage{color}
\usepackage{url}
\usepackage{hyperref}
\usepackage{bm}
\usepackage{bbm}
\usepackage{times}

\usepackage{enumitem}

\newtheorem{thm}{Theorem}[section]

\newtheorem{lemma}[thm]{Lemma}

\newtheorem{claim}[thm]{Claim}

\newtheorem{definition}[thm]{Definition}
\newtheorem{remark}[thm]{Remark}

\newtheorem{fact}[thm]{Fact}

\newcommand\card[1]{\left| {#1} \right|}
\newcommand\sett[2]{\left\{ \left. #1 \;\right\vert #2 \right\}}

\newcommand\Prob[2]{{\Pr_{#1}\left[ {#2} \right]}}

\newcommand\Expect[2]{{\mathop{\mathbb{E}}_{#1}\left[ {#2} \right]}}

\newcommand\eps{\varepsilon}

\renewcommand\geq{\geqslant}
\renewcommand\leq{\leqslant}

\newcommand{\rom}[1]{\uppercase\expandafter{\romannumeral #1\relax}}

\title{Parallel Repetition for $k$-Player XOR Games}
\author{Mark Braverman
\thanks{Department of Computer Science, Princeton University.
Research supported in part by the NSF Alan T. Waterman Award, Grant No. 1933331, a Packard Fellowship in Science and Engineering, and the Simons Collaboration on Algorithms and Geometry. }
\and Dor Minzer \thanks{Department of Mathematics, Massachusetts Institute of Technology. Supported by NSF CCF award 2227876, NSF CAREER award 2239160 and NSF DMS award 2349024.}}
\date{\vspace{-5ex}}
\begin{document}
\maketitle
\begin{abstract}
In a $k$-XOR game $\Psi$ over $\mathbb{F}_p$, a verifier samples challenges $(x_1,\ldots,x_k)\sim\mu$ where $\mu$ is a distribution over $\Sigma_1\times\ldots\times \Sigma_k$, sends $x_i$ to player $i$ and receives an answer $f_i(x_i)\in\mathbb{F}_p$, and accepts if and only if $f_1(x_1)+\ldots+f_k(x_k)=t(x_1,\ldots,x_k)$, where $t\colon \Sigma_1\times\ldots\times \Sigma_k\to\mathbb{F}_p$ is a fixed map. We show that if $\Psi$ has value strictly less than $1$, then the value of the $n$-fold repeated game $\Psi^{\otimes n}$ is at most $2^{-cn}$ where $c = c(\Psi)>0$ is a fixed constant. This extends a previous result of [Bhangale-Braverman-Khot-Liu-Minzer] showing exponential decay for some $3$-XOR games.

Our proof uses a relaxation of the value of games studied in the quantum setting, called the commuting value of a game. We show a transformation that maps a $k$-player game $\Psi$ with value strictly less than $1$ to a $2k$-player game $\Psi'$ with commuting value less than $1$, such that $\mathsf{val}(\Psi^{\otimes n})\leq \mathsf{val}((\Psi')^{\otimes n})$. We prove that $\mathsf{val}(\Psi')<1$ has a direct proof using nothing more than the Cauchy-Schwarz inequality, which we are able to tensorize to get that $\mathsf{val}((\Psi')^{\otimes n})$ vanishes exponentially fast in $n$.
\end{abstract}

\section{Introduction}
Parallel repetition, introduced by~\cite{FortnowRompelSipser1994}, is a powerful tool that amplifies the soundness guarantee of interactive protocols. A $k$-player-$1$-round game $\Psi$ consists of a verifier and $k$ non-communicating provers, where player $i$ has a set of viable challenges $\Sigma_i$ and viable responses $A_i$. The game is defined by a probability distribution $\mu$ over $\Sigma_1\times\ldots\times\Sigma_k$ and a predicate $V\colon \prod\limits_{i=1}^{k} \Sigma_i\times \prod\limits_{i=1}^{k} A_i\to \{0,1\}$, and proceeds as follows:
\begin{enumerate}
    \item The verifier samples challenges $(x_1,\ldots,x_k)\sim \mu$, and sends $x_i$ to player $i$.
    \item For each $i\in [k]$, player $i$ chooses a response $a_i$ based on their challenge $x_i$ and sends it to the verifier.
    \item The verifier accepts if and only if $V(x_1,\ldots,x_k,a_1,\ldots,a_k)=1$
\end{enumerate}
The value of $\Psi$, denoted by $\mathsf{val}(\Psi)$, is defined as the maximum acceptance probability of the verifier under all players' strategies. The above framework captures many interesting settings in theoretical computer science and particularly in the theory of probabilistically checkable proofs; for example, the PCP theorem~\cite{FGLSS,AS,ALMSS} has the following equivalent formulation: there is a polynomial time reduction that maps a $3$-CNF formula $\varphi$ of size $n$ to a $2$-player-$1$-round game $\Psi$ with question sets of size $n^{O(1)}$, and response sets of size $O(1)$, such that
\begin{enumerate}
    \item \textbf{Completeness:} if $\varphi$ is satisfiable, then $\mathsf{val}(\Psi)=1$.
    \item \textbf{Soundness:} if $\varphi$ is unsatisfiable, then $\mathsf{val}(\Psi)\leq 0.99$.
\end{enumerate}
This gap between the completeness and soundness cases is often too mild, and one would like to amplify it. The main goal of parallel repetition is precisely to achieve that, while at the same time keeping the same number of players and the same number of rounds. For a parameter $n$, the $n$-fold repeated game $\Psi^{\otimes n}$ is defined as follows (we describe the verifier, and the question and answer sets are implicit):
\begin{enumerate}
    \item The verifier samples challenges $(x_1(j),\ldots,x_k(j))\sim \mu$ independently for each $j=1,\ldots,n$, and sends $x_i = (x_i(1),\ldots,x_i(n))$ to player $i$.
    \item For each $i\in [k]$, player $i$ chooses a response $(a_i(1),\ldots,a_i(n))$ based on $x_i$ and sends it to the verifier.
    \item The verifier accepts if and only if $V(x_1(j),\ldots,x_k(j),a_1(j),\ldots,a_k(j))=1$
    for all $j=1,\ldots,n$.
\end{enumerate}
It is easy to see that if $\mathsf{val}(\Psi)=1$, then $\mathsf{val}(\Psi^{\otimes n})=1$. It is much less clear what happens if $\mathsf{val}(\Psi)<1$, and the parallel repetition conjecture of~\cite{FortnowRompelSipser1994} predicts that $\mathsf{val}(\Psi^{\otimes n})\leq 2^{-\Omega_{\Psi}(n)}$.
\subsection{Prior Work}
The best known result for general games is due to~\cite{Verbitsky1994}, and asserts that $\lim_{n\rightarrow \infty}\mathsf{val}(\Psi^{\otimes n})=0$ (see also~\cite{Mittal2025}). Much better results are known for more restricted classes of games, as we discuss below.
\vspace{-2ex}
\paragraph{$2$-Player Games:} the case of $2$-player games, which is very important due to its connection to $2$-query PCPs, was heavily studied in the literature, and is by now fully understood. The work~\cite{Raz1998} was the first to establish that for a general $2$-player game $\Psi$ with value less than $1$, $\mathsf{val}(\Psi^{\otimes n})$ indeed decays exponentially fast in $n$. Follow up works~\cite{Holenstein2009,Rao2011,DinurSteurer2014,BravermanGarg2015} improved upon it, getting nearly tight decay rates for general $2$-player games as well as for the important class of projection $2$-player games. With the exception of~\cite{DinurSteurer2014}, all proofs use an approach based on information theory.
\vspace{-2ex}
\paragraph{$k$-Player Games:} the case of $k$-player games for $k\geq 3$ is considerably more challenging. The work~\cite{DinurHarshaVenkatYuen2017} identified a subclass of $k$-player games, called connected games, for which the information theoretic approach applies and established an exponential decay. They suggested that strong correlations between the questions of the players is a bottleneck for their technique, and suggested to study a concrete $3$-player game known as the GHZ game. The works~\cite{HolmgrenRaz2020,GirishEtAl2021} showed that the $n$-fold repeated GHZ game decays at least as quickly as $\frac{1}{n^c}$ for some constant $c>0$, and later generalized their argument to all $3$-player games with binary questions~\cite{GirishMittalRazZhan2022,GirishEtAl2022}. Subsequent works discovered a connection between the GHZ game~\cite{BravermanKhotMinzer2023} and inverse theorems for CSPs~\cite{BhangaleKhotMinzer2025I,BhangaleKhotMinzer2025II,BhangaleKhotMinzer2023III,BhangaleKhotMinzer2024IV,BhangaleKhotMinzer2025V,BhangaleKhotLiuMinzer2024VI,BhangaleKhotLiuMinzer2024VII}, leading to an exponential decay rate for many $3$-XOR games~\cite{BhangaleEtAl2025} and for $k$-player projection games~\cite{BhangaleEtAl2024}. Other words are~\cite{BhangaleEtAl2026}, who obtained decay rates of the form $\frac{1}{\log^{(C)} n}$ (where $\log^{(C)}$ is the $C$-fold iterated log) for $k$-player games in which the distribution over questions has no Abelian embeddings, and~\cite{LiuLovettMittal2026}, who obtained concentration bounds for games with GHZ-questions.

\vspace{-1ex}
\paragraph{Our starting point:} the current paper was born in an attempt to give a more direct and transparent proof for the result of~\cite{BravermanKhotMinzer2023} for the GHZ game, and we begin by explaining their approach. In the GHZ game the verifier samples $3$ bits $x_1,x_2,x_3\in \mathbb{F}_2$ conditioned on $x_1+x_2+x_3=0$, sends a single bit to each player, receives answers $a_1,a_2,a_3$ and accepts if and only if $a_1+a_2+a_3 = x_1\lor x_2\lor x_3\pmod{2}$. In the $n$-fold parallel repetition of the GHZ game, the players receive $x_1,x_2,x_3\in \mathbb{F}_2^{n}$ where $x_1+x_2+x_3 = \vec{0}$, respond with $f_1(x_1),f_2(x_2),f_3(x_3)\in \mathbb{F}_2^n$, and win if $f_1(x_1)+f_2(x_2)+f_3(x_3) = x_1\lor x_2\lor x_3$, where the OR is taken coordinate-wise. Algebraically, the winning probability can be expressed as
\begin{align}\label{eq:intro}
&\Expect{(x_1,x_2,x_3)\sim \mu^{\otimes n}}
{
\prod\limits_{j=1}^{n}\frac{1}{2}
(
(-1)^{f_1(x_1)(j)+f_2(x_2)(j)+f_3(x_3)(j)-x_1(j)\lor x_2(j)\lor x_3(j)}+1)}\notag\\
&=\frac{1}{2^{n}}
\sum\limits_{S\subseteq [n]}\Expect{(x_1,x_2,x_3)\sim\mu^{\otimes n}}
{(-1)^{f_{1,S}(x_1)+f_{2,S}(x_2)+f_{3,S}(x_3)-
\sum\limits_{j\in S}x_1(j)\lor x_2(j)\lor x_3(j)}},
\end{align}
where $f_{1,S}(x_1) = \sum\limits_{j\in S} f_1(x_1)(j)$ and $f_{2,S}$, $f_{3,S}$ are defined similarly. Thus, the problem of studying the value achieved by $f_1,f_2,f_3$ is reduced to understanding the direst sum value achieved by them: show that for a typical $S\subseteq [n]$ it holds that
\[
\card{\Expect{(x_1,x_2,x_3)\sim\mu^{\otimes n}}
{(-1)^{f_{1,S}(x_1)+f_{2,S}(x_2)+f_{3,S}(x_3)-
\sum\limits_{j\in S}x_1(j)\lor x_2(j)\lor x_3(j)}}}\leq 2^{-\Omega(n)}.
\]
Alas, this turns out to be false, and for each $S$ one can design $g_1,g_2,g_3\colon \mathbb{F}_2^n\to\mathbb{F}_2$ for which
\[
\card{\Expect{(x_1,x_2,x_3)\sim\mu^{\otimes n}}
{(-1)^{g_{1}(x_1)+g_{2}(x_2)+g_{3}(x_3)-
\sum\limits_{j\in S}x_1(j)\lor x_2(j)\lor x_3(j)}}}
\geq \Omega(1).
\]
The reason for that is that the GHZ game has an Abelian embedding, and we define this notion more generally below.
\begin{definition}
    A $k$-player game $\Psi = (\Sigma_1,\ldots,\Sigma_k,A_1,\ldots,A_k,\mu,V)$ is called a $k$-XOR game if for some prime $p\in\mathbb{N}$, we have $A_1=\ldots=A_k=\mathbb{F}_p$ and there exists $t\colon \Sigma_1\times\ldots\times\Sigma_k\to\mathbb{F}_p$ such that
    \[
    V(x_1,\ldots,x_k,a_1,\ldots,a_k) = 1
    \iff
    a_1+\ldots+a_k = t(x_1,\ldots,x_k).
    \]
\end{definition}
In words, a $k$-XOR game is a game in which the predicate the verifier checks is linear over the responses of the players.
\begin{definition}\label{def:Abelian_embed}
    We say a $k$-XOR game has an Abelian embedding if there are functions $\sigma_i\colon \Sigma_i\to \mathbb{C}$ with absolute value $1$ pointwise such that for all $(x_1,\ldots,x_k)\in\mathsf{supp}(\mu)$ we have that
    \[
    \prod\limits_{i=1}^{k}\sigma_i(x_i) = \omega_p^{t(x_1,\ldots,x_k)},
    \]
    where $\omega_p = e^{\frac{2\pi\mathbf{i}}{p}}$ is the primitive root of unity of order $p$.
\end{definition}
In words, an Abelian embedding is a perfect solution to the $k$-player, except that the ``answers'', $\sigma_i(x_i)$, do not come from the set of viable answers. Nevertheless, existence of such non-valid solution presents an issue to the study of the direct sum value. For instance, for the GHZ game one has the $(\mathbb{Z}_4,+)$ structure
\[
(-1)^{x_1\lor x_2\lor x_3} = \mathbf{i}^{x_1}\mathbf{i}^{x_2}\mathbf{i}^{x_3},
\]
for all valid tuples of challenges $x_1,x_2,x_3$, and this can be used to construct $g_1,g_2,g_3$ as above. Due to this reason, the work~\cite{BravermanKhotMinzer2023} takes a different approach, and uses additive combinatorics tools to argue that any successful strategy for the $n$-fold repeated game must make use of this $\mathbb{Z}_4$-structure, and then analyze it directly to show that it does not work.

\subsection{Our Results}
Motivated by the above discussion, we tried to find a simpler proof for the exponential decay of the value of the $n$-fold repeated game of the GHZ game. We found such a proof, and upon further investigation as we discuss in Section~\ref{sec:techniques}, we managed to make it work for all $k$-XOR games, giving our main result below.
\begin{thm}\label{thm:main}
    Suppose that $\Psi$ is a $k$-XOR game and that $\mathsf{val}(\Psi)< 1$. Then there exists $\eps=\eps(\Psi)>0$ such that  $\mathsf{val}(\Psi^{\otimes n})\leq (1-\eps)^n$ for all $n\in\mathbb{N}$.
\end{thm}
Theorem~\ref{thm:main} improves upon the result of~\cite{BhangaleEtAl2025} in two ways. First, it applies to $k$-XOR games for all $k\geq 2$, whereas the result of~\cite{BhangaleEtAl2025} applies only to $k=3$, and secondly it requires no assumption on the distribution over questions $\mu$, whereas the result of~\cite{BhangaleEtAl2025} requires that the game has no $(\mathbb{Z},+)$ embedding.
\begin{remark}
The assertion in Theorem~\ref{thm:main} holds for the more general class of $k$-XOR games over any finite Abelian group, and we chose to state it only for games over prime fields for the sake of presentation. The key difference is that one generalizes the BLR-trick, as discussed in Section~\ref{sec:techniques}, from using characters over $\mathbb{F}_p$ to using characters over the Abelian group.
\end{remark}

Unlike previous proofs of exponential parallel repetition
for the GHZ game and $3$-XOR games, which relied on
additive combinatorics~\cite{BravermanKhotMinzer2023,BhangaleEtAl2025},
our proof relies on Cauchy--Schwarz inequalities and
their tensorization.
This raises the question of whether all parallel repetition
results can be proved using such techniques.
Even for $2$-player games, it remains unclear whether
the analytic approach of Dinur and Steurer~\cite{DinurSteurer2014}
can recover the general-game parallel repetition theorems
of Raz~\cite{Raz1998}, Holenstein~\cite{Holenstein2009},
and Braverman--Garg~\cite{BravermanGarg2015}, as well as
the concentration bound of Rao~\cite{Rao2011}.

\subsection{Techniques}\label{sec:techniques}
In this section we discuss our techniques, and it is first useful to discuss the result of~\cite{BhangaleEtAl2025}. Both approaches are motivated from the notion of Abelian embeddings, but they take very different routes.
\subsubsection{The Previous Approach to XOR Games}
 The proof of~\cite{BhangaleEtAl2025} has two cases, depending on if $t$ has an Abelian embedding or not.
\begin{enumerate}
    \item In case $t$ has an Abelian embedding, the proof uses techniques from additive combinatorics to essentially argue about the structure any successful strategy must have, extending the argument of~\cite{BravermanKhotMinzer2023}, and then directly showing that such strategies do not work.
    \item In case  $t$ has no Abelian embedding, the proof relies on a tensorization argument along the lines of~\cite{BhangaleKhotMinzer2025II}. It shows that for each $S$ of size $\Omega(n)$, the corresponding summand on the right hand side of~\eqref{eq:intro} is at most $2^{-\Omega(n)}$. Morally, the idea here is that the fact that $t$ has no Abelian embedding can be interpreted as the existence of $\delta>0$ such that for each $g_i\colon \Sigma_i\to\mathbb{C}$ with absolute value $1$, we have that
    \[
    \card{\Expect{(x_1,x_2,x_3)\sim\mu}{g_1(x_1)g_2(x_2)g_3(x_3)\omega_p^{t(x_1,x_2,x_3)}}}\leq 1-\delta.
    \]
    This can be thought of as a ``base case'' that can then be amplified.
\end{enumerate}
We note that the GHZ game falls in the first item above, which in some sense looks more ad-hoc than the second item. Indeed, our starting point was to investigate to what extent can one use tensorization ideas as in the second item to prove parallel repetition theorems.
\subsubsection{Our Approach}\label{sec:approach}
In this section we discuss our approach for the proof of Theorem~\ref{thm:main}, and we describe it in the way we arrived at it.
\paragraph{Cauchy-Schwarz Proofs:} drawing intuition from~\cite{BhangaleEtAl2025}, our first goal was to prove that any $k$-XOR game with no Abelian embedding has an exponentially decaying value under parallel repetition, and furthermore that there should be a ``direct'' proof of that fact. We were able to do it for $k=2$, where the idea is that if $\mathsf{val}(\Psi)<1$, then $\Psi$ can be viewed as a system of linear equations that has a contradiction. One then shows that a general such contradiction can be converted to a structured contradiction that can then be evidenced by several applications of the Cauchy-Schwarz inequality, and thus tensorizes. A simple demonstration is given by the inner-product game, where the distribution over questions $\mu$ has full support over $\mathbb{F}_2\times\mathbb{F}_2$, the players respond with $f(x)$, $g(y)$ respectively, and the verifier accepts if $f(x)+g(y) = x\cdot y$. Then a direct way to show that the value of this game is less than $1$ is by first noting that
\[
\card{\Expect{(x,y)\sim\mu}{(-1)^{f(x)}(-1)^{g(y)}(-1)^{xy}}}^2
\leq
\Expect{\substack{(x,y)\sim \mu\\ (x',y')\sim \mu~|~x'=x}}{(-1)^{g(y)+g(y')}(-1)^{x(y+y')}},
\]
and we denote the distribution over $(x,y,y')$ by $\mu'$. We have that
\[
\card{\Expect{\substack{(x,y,y')\sim \mu'}}{(-1)^{g(y)+g(y')}(-1)^{x(y+y')}}}^2
\leq
\Expect{\substack{(x,y,y')\sim \mu'\\ (x',y'',y''')\sim \mu'~|~(y'',y''')=(y,y')}}{(-1)^{(x+x')(y+y')}},
\]
which is less than $1$ as $(x+x')(y+y')$ is not constant over the support of the last distribution. Furthermore, this proof ``tensorizes'', in the sense that pretty much the same set of inequalities hold for the $n$-fold repeated game and give a bound of $2^{-\Omega(n)}$.\footnote{For general $2$-XOR games, a contradiction can be thought of as a list of integer coefficients $a_{i,j}$ to multiply the equation that contains variables $i,j$, such that $\sum_{i}a_{i,j}=\sum\limits_{j}a_{i,j}=0$ for all $i,j$ but $\sum\limits_{i,j}a_{i,j}t(i,j)\neq 0$. Using the Birkhoff–von Neumann theorem one can show that the coefficients matrix $A$ can be written as $A = \sum\limits_{t}\lambda_t C_t$ where $C_t$ has entries $\{-1,0,1\}$, and it is supported on an alternating row/column walk with sum $0$ inside the support of $A$. Thus, there is $C_t$ with $\sum\limits_{i,j}C_t(i,j)t(i,j)\neq 0$ and each row/column sum in $C_t$ is $0$, which corresponds to addition/subtraction of equations $e_1,\ldots,e_{m}$, where the second variable of $e_i$ gets canceled by the first variable in $e_{i+1}$. Such a contradiction can be converted to a sequence of applications of Cauchy-Schwarz.}
Next, we were able to find a non-trivial direct proof for some $3$-player games, such as the game considered in~\cite{braverman2026parity}. Therein, the verifier samples a permutation $\pi\in S_3$, sends $\pi(1),\pi(2),\pi(3)$ to player $1$, player $2$, player $3$ respectively, they respond with signs $a_1,a_2,a_3$ and the verifier accepts if and only if
$a_1a_2a_3 = \mathsf{sign}(\pi)$. While the result of~\cite{BhangaleEtAl2025} holds for this game and thus implies it has exponentially small value, it turns out to have a much more direct proof using only the Cauchy-Schwarz inequality (in Appendix~\ref{app:levi-civita-gluing} we give a different direct proof that follows a more general recipe).

This led us to suspect that any $k$-XOR game with no Abelian embedding should have an exponential decay under parallel repetition, and furthermore that this fact should be provable by several applications of the Cauchy-Schwarz inequality. Though ultimately this direct Cauchy-Schwarz approach failed, we were unable to resolve it for a while, and instead turned our attention to games with Abelian embeddings.
\vspace{-1ex}
\paragraph{The BLR-trick:} motivated by the suspicion above, one may try to reduce a $k$-XOR game $\Psi$ that has an Abelian embedding to an XOR game $\Psi'$ that has no Abelian embedding, and this turns out to be possible. The hope is that $\mathsf{val}(\Psi)<1$ means that while $\Psi$ may have an Abelian embedding, it cannot be one that only uses roots of unity of order $p$. Therefore, one can design a game $\Psi'$ that does some local correction that heavily uses the $\mathbb{F}_p$-arithmetic structure, thereby being completely incompatible with any other arithmetic structure. A natural candidate inspired by~\cite{BLR} is the game $\Psi'$, which has $2k$ players labeled by $1,\ldots,k$ and $1',\ldots,k'$. The question set of players $i$ and $i'$ are $\Sigma_i' = \Sigma_i\times \mathbb{F}_p$, and their responses are in $\mathbb{F}_p$.
The verifier proceeds as follows:
\begin{enumerate}
    \item Sample $(x_1,\ldots,x_k)\sim \mu$, sample $\gamma\in\mathbb{F}_p$ and $\alpha_1,\ldots,\alpha_k\in\mathbb{F}_p$ uniformly and independently.
    \item Send question $(x_i,\alpha_i)$ to player $i$ and $(x_i,\alpha_i+\gamma)$ to player $i'$.
    \item Upon receiving the answers $f_i(x_i,\alpha_i)$ and $f_{i'}(x_i,\alpha_i+\gamma)$, the verifier accepts if and only if
    \[
    \sum\limits_{i=1}^{k}f_{i'}(x_i,\alpha_i+\gamma)+f_{i}(x_i,\alpha_i) = \gamma t(x_1,\ldots,x_k).
    \]
\end{enumerate}
Note that if $g_i\colon \Sigma_i\to \mathbb{F}_p$ are strategies for $\Psi$, then the natural strategies $f_i(x_i,\alpha) = -\alpha g_i(x_i)$
and $f_{i'}(x_i,\alpha) = \alpha g_{i}(x_i)$ are valid strategies, and one gets that $\mathsf{val}(\Psi)\leq \mathsf{val}(\Psi')$ and $\mathsf{val}(\Psi^{\otimes n})\leq \mathsf{val}((\Psi')^{\otimes n})$. On the other direction, it is also quite easy to see that if $\mathsf{val}(\Psi)<1$, then $\mathsf{val}(\Psi')<1$. The main feature of this transformation is that it in fact achieves the following stronger property:
\begin{fact}\label{fact:BLR_trick_basic}
    Suppose that $\Psi$ satisfies $\mathsf{val}(\Psi)<1$. Then $\Psi'$ has no Abelian embedding.
\end{fact}
\begin{proof}[Proof sketch]
Suppose $\Psi'$ had an Abelian embedding, and take functions $\sigma_{i},\sigma_{i'}$ such that
\[
\prod\limits_{i}\sigma_i(x_i,\alpha_i)\sigma_{i'}(x_i,\alpha_i+\gamma) = \omega_p^{\gamma t(x_1,\ldots,x_k)}.
\]
Fixing
$(x_1,\ldots,x_k)$, we see that the right hand side is a character $\gamma\rightarrow \omega_{p}^{c\gamma}$, and it does not depend on any of the $\alpha_i$. It follows that the functions $\sigma_{i}(x_i,\alpha), \sigma_{i'}(x_i,\alpha)$ must each be a constant multiple of a character, and the characters must be the conjugates of each other, say $\theta_{i,x_i}\omega_p^{-a_i(x_i)\cdot \alpha}$ and $\theta_{i',x_i}\omega_p^{a_{i}(x_i)\cdot \alpha}$. Plugging this into the embedding equation gives that the constants must multiply to $1$, and that $\omega_{p}^{\gamma(a_1(x_1)+\ldots+a_k(x_k))}=\omega_p^{\gamma t(x_1,\ldots,x_k)}$, and so $a_1(x_1)+\ldots+a_k(x_k) = t(x_1,\ldots,x_k)$ and thus $\mathsf{val}(\Psi)=1$.
\end{proof}
\vspace{-1ex}
\paragraph{Combining the BLR-trick and Cauchy-Schwarz Proofs:} Fact~\ref{fact:BLR_trick_basic} together with the Cauchy-Schwarz plan suggest a way to prove Theorem~\ref{thm:main}. Using it, we were able to produce a proof for the exponential decay of the $n$-fold repeated GHZ game, but not for general XOR games. In fact, we realized that the Cauchy-Schwarz plan above does not actually work, and there are games $\Psi$ that have no Abelian embedding and yet the direct-sum statement for them fails. Motivated by this fact we tried to understand what made the Cauchy-Schwarz proof work for the GHZ game after the BLR-trick, and it turns out that the BLR-trick satisfies a feature much stronger than Fact~\ref{fact:BLR_trick_basic} that we present next.

\subsubsection{Commuting Embeddings}
Below is a generalization of the notion of Abelian embeddings from Definition~\ref{def:Abelian_embed} from~\cite{BWHK23}, that turns out to be crucial in our setting (a related notion has been studied in~\cite{BeneWattsHelton2023}).
\begin{definition}\label{def:commuting_embedding}
    Let $\Psi$ be a $k$-XOR game. We say that $\Psi$ has a commuting embedding if there is a complex Hilbert space $V$ (not necessarily finite dimensional) and unitaries $U_{i}(x_i)$ over $V$ for all $i\in [k]$ and $x_i\in\Sigma_i$ such that
    \begin{enumerate}
        \item For all $i\neq j$, $x_i\in\Sigma_i$ and $x_j\in \Sigma_j$ it holds that $U_{i}(x_i)$ and $U_{j}(x_j)$ commute.
        \item There exists a non-zero vector $v\in V$ such that for all $(x_1,\ldots,x_k)\in\mathsf{supp}(\mu)$ it holds that
        \[
        U_1(x_1)\cdots U_k(x_k)v = \omega_p^{t(x_1,\ldots,x_k)}v.
        \]
    \end{enumerate}
\end{definition}
It is easy to note that Definition~\ref{def:Abelian_embed} is the same as Definition~\ref{def:commuting_embedding} in the case that $\mathsf{dim}(V)=1$, hence not having a commuting realization is a stronger property. It turns out that the BLR-trick achieves this stronger property, whose proof we defer to Section~\ref{sec:proofs}.
\begin{lemma}\label{lem:BLR_trick_strong}
    Suppose that $\Psi$ satisfies $\mathsf{val}(\Psi)<1$. Then $\Psi'$ has no commuting embedding.
\end{lemma}
Lemma~\ref{lem:BLR_trick_strong} gives new hope into the plan outlined in Section~\ref{sec:approach}: it now suffices to prove that any game that does not have any commuting embedding has a structured contradiction that we can implement via a Cauchy-Schwarz proof (and which we can then tensorize). This turns out to be true, and is the content of~\cite[Theorems 5.10, 5.12]{BWHK23} (earlier established by~\cite[Theorem 6.1]{BWHKN19} for games over $\mathbb{F}_2$), and combining that together with a simple tensorization bound yields Theorem~\ref{thm:main}.

\paragraph{Statement of AI use:} The manuscript is part of an ongoing effort to obtain general parallel repetition theorems using tensorizations of Cauchy-Schwarz. We used ChatGPT 5.5, 5.6 and 6 to test our earlier specialized Cauchy-Schwarz attempt, getting direct proofs for the Levi-Cevita game and the GHZ game (after the BLR trick).

ChatGPT 6 helped us find the relevant quantum literature, which directed our search towards the present Cauchy-Schwarz approach. We instructed ChatGPT 6 to find a Cauchy-Schwarz proof for the GHZ game after the BLR trick has been applied to it, for which it was successful. In an attempt to try to see why that succeeded, we asked the model to prove that if a game has quantum value less than $1$ then parallel repetition holds. The model was able to solve this problem by using results from~\cite{BWHK23}, which led us to Lemma~\ref{lem:BLR_trick_strong}.

\section{On the BLR-trick,  Commuting Embeddings and
Ordered Contradictions}\label{sec:proofs}
In this section we prove Lemma~\ref{lem:BLR_trick_strong} and then show the relation between commuting embeddings and a form of structured constraints called ordered contradictions.
\subsection{Proof of Lemma~\ref{lem:BLR_trick_strong}}
Suppose towards contradiction that $U_{1},U_{1'},\ldots,U_{k},U_{k'}$ are a commuting embedding of $\Psi'$ over the vector space $V$, and fix a non-zero $v\in V$ as in Definition~\ref{def:commuting_embedding} so that for all $(x_1,\ldots,x_k)\in\mathsf{supp}(\mu)$, $\alpha_1,\ldots,\alpha_k,\gamma\in\mathbb{F}_p$ it holds that
\begin{equation}\label{eq:BLR_trick}
\prod\limits_{i=1}^{k}U_{i'}(x_i,\alpha_i+\gamma)U_i(x_i,\alpha_i) v = \omega_p^{\gamma t(x_1,\ldots,x_k)} v.
\end{equation}
\paragraph{Step 1: symmetrizing player $i$ and $i'$.} The goal of the first step is to show that the operator of player $i'$ should (morally) be the conjugate transpose of the operator of player $i$. Towards this end define $D_{i}(x,\alpha) =
U_{i}(x,\alpha)U_{i}(x,0)^{*}$,
$E_i(x,\alpha) = U_{i'}(x,\alpha)U_{i'}(x,0)^{*}$.
Fixing $(x_1,\ldots,x_k)$, and plugging $\gamma=0$ into~\eqref{eq:BLR_trick} we get
$\prod\limits_{i=1}^{k}U_{i'}(x_i,\alpha_i)U_i(x_i,\alpha_i) v = v$, which has the following implications:
\begin{enumerate}
    \item We have that $\prod\limits_{i=1}^{k}E_i(x_i,\alpha_i)D_i(x_i,\alpha_i) v = v$, so for each $j$, picking $\alpha_i=0$ for all $i\neq j$ gives that $E_j(x_j,\alpha_j)D_j(x_j,\alpha_j) v= v$. By commutation we get that
    $D_j(x_j,\alpha_j)v=E_j(x_j,\alpha_j)^{*}v$.
    \item Using this together with~\eqref{eq:BLR_trick} gives
    \[
\prod\limits_{i=1}^{k}U_{i'}(x_i,\alpha_i+\gamma)U_i(x_i,\alpha_i)
\left(\prod\limits_{i=1}^{k}U_{i'}(x_i,\alpha_i)U_i(x_i,\alpha_i)\right)^{*}v
=\omega_p^{\gamma t(x_1,\ldots,x_k)} v,
\]
and expanding the left hand side and using the commutation property gives that
\[
\prod\limits_{i=1}^{k}E_i(x_i,\gamma+\alpha_i)E_i(x_i,\alpha_i)^{*}v=\omega_p^{\gamma t(x_1,\ldots,x_k)} v.
\]
\end{enumerate}

\paragraph{Step 2: arguing that $E_i$ are representations.}
Define $W$ to be the closure of
\[
\mathsf{span}\left\{D_{i_1}(x_{1},\alpha_{1})^{n_1}\cdots D_{i_m}(x_{m},\alpha_{m})^{n_m} v~|~m\geq 0, ~\forall j\in[m] ~ x_j\in \Sigma_{i_j}, \alpha_j\in\mathbb{F}_p, n_j\in \{-1,1\}, i_j\in [k]\right\}.
\]
By the commutation property we get that $W$ is an invariant space for each $D_i$. Next, using commutativity again and $D_j(x_j,\alpha_j)v = E_j(x_j,\alpha_j)^{*}v$, it follows that
\[
D_{i_1}(x_{1},\alpha_{1})^{n_1}\cdots D_{i_m}(x_{m},\alpha_{m})^{n_m} v
=E_{i_m}(x_m,\alpha_m)^{-n_m}\cdots E_{i_1}(x_1,\alpha_1)^{-n_1}v,
\]
so the space $W$ is the same if we replaced $D$'s with $E$'s. Thus, $W$ is also an invariant space for each $E_i$. Lastly, it follows that for all $w\in W$,
\[
\prod\limits_{i=1}^{k}E_i(x_i,\gamma+\alpha_i)E_i(x_i,\alpha_i)^{*} w = \omega_p^{\gamma t(x_1,\ldots,x_k)}w.
\]

We thus restrict all of the $D$ and $E$ operators to $W$, and get that
\begin{equation}\label{eq:BLR_trick_2}
\prod\limits_{i=1}^{k}E_i(x_i,\gamma+\alpha_i)E_i(x_i,\alpha_i)^{*} = \omega_{p}^{\gamma t(x_1,\ldots,x_k)} I_W.
\end{equation}
Fixing $\gamma$ and  $\alpha_1,\ldots,\alpha_k,\alpha_1',\ldots,\alpha_k'\in\mathbb{F}_p$ that only disagree on coordinate $j$
we get
\[
E_j(x_j,\gamma+\alpha_j)
E_j(x_j,\alpha_j)^{*}
=
E_j(x_j,\gamma+\alpha_j')
E_j(x_j,\alpha_j')^{*}.
\]
Picking $\alpha_j'=0$ using $E_j(x_j,0)=I_W$ and re-arranging gives that
$E_j(x_j,\gamma+\alpha_j)
 = E_j(x_j,\gamma)E_j(x_j,\alpha_j)$, so for each $x_j\in\Sigma_j$, the pair $(\alpha\rightarrow E_j(x_j,\alpha),W)$ is a representation of $(\mathbb{F}_p,+)$. In particular, we get that $E_{j}(x_j,\alpha)^{p} = I_W$ for each $j$, $x_j$ and $\alpha$.

\paragraph{Step 3: getting a classical strategy for $\Psi$:} we note that for $i\neq i'$, the operators $E_{i}(x_i,\alpha_i)$, $E_{i'}(x_{i'},\alpha_{i'})$ commute with each other by assumption on the $U$. Plugging $\alpha_i=0$ to~\eqref{eq:BLR_trick_2}, we see that
$E_i(x_i,\alpha_i)$ can be expressed as a product of $E_{i'}^{*}$ for $i'\neq i$, so $E_i(x_i,\alpha_i)$ and $E_i(x_i',\alpha_i')$ also commute with each other. Thus, the operators $E$ all commute with each other, and $E^{p} = I_W$ for each one of them, and it follows that  $\mathsf{dim}(W)<\infty$. Thus, we think of $E_{j}(x_j,\alpha)$ as finite matrices henceforth.

It follows that all of the $E_i$ matrices are simultaneously diagonalizable, and we may find a common eigenvector $w\in W$.
It follows that $(\alpha\rightarrow E_j(x_j,\alpha),\mathsf{span}(w))$ is an irreducible representation, and since $(\mathbb{F}_p,+)$ is Abelian we conclude that
$E_j(x_j,\alpha) w = \chi_{j, x_j}(\alpha)w$ where $\chi_{j,x_j}\colon \mathbb{F}_p\to\mathbb{C}$ is a character. Writing $\chi_{j,x_j}(z) = \omega_p^{\sigma_j(x_j)z}$ for some $\sigma_j\colon \Sigma_j\to \mathbb{F}_p$ and plugging into~\eqref{eq:BLR_trick_2} with $\alpha_i=0$ gives
\[
\prod\limits_{i=1}^{k}\omega_p^{\sigma_i(x_i)\gamma} w= \omega_p^{\gamma t(x_1,\ldots,x_k)} w,
\]
so $\sum\limits_{i=1}^{k}\sigma_i(x_i) = t(x_1,\ldots,x_k)$, concluding the proof.

\subsection{From No Commuting Embeddings to Ordered Contradictions}
In this section we prove that if a $k$-XOR game has no commuting embeddings, then it admits a structured form of contradictions, which we call an \emph{ordered contradiction}, which informally is a collection of constraints and signing of them that lead to a combinatorial contradiction.
To define it formally we need to extend the class of $k$-XOR games a little bit to allow having constraints which have the same left hand side but different right hand side:
\begin{definition}
    A generalized $k$-XOR game is defined by sets of questions $\Sigma_1,\ldots,\Sigma_k$ and a distribution $\mu$ over $\Sigma_1\times\ldots\times\Sigma_k\times\mathbb{F}_p$. The verifier samples $(x_1,\ldots,x_k,t)\sim \mu$, sends $x_i$ to player $i$, receives answer $a_i\in\mathbb{F}_p$ and accepts if and only if $a_1+\ldots+a_k = t$.

    We also allow $k=0$, in which case the verifier
simply samples $t\in\mathbb{F}_p$ and accepts
if and only if $t=0$.
\end{definition}
We switch between the $t\in\mathbb{F}_p$ of right hand side of the challenge and its complex analog $\lambda=\omega_p^{t}$, in which case the constraint becomes $\omega_p^{a_1}\cdots\omega_p^{a_k} = \lambda$.
\begin{definition}\label{def:ordered_contradiction}
Let $\Psi$ be a generalized $k$-XOR game. An ordered contradiction of length $N$ for $\Psi$ is a sequence of challenges $(e_1,t_1),\ldots,(e_N,t_N)\in \mathsf{supp}(\mu)$ and signs $s_1,\ldots,s_N\in \{-1,1\}$ such that:
\begin{enumerate}
    \item The right hand side does not trivialize, in the sense that $\prod\limits_{j=1}^{N}{\omega_p}^{s_jt_j}\neq 1$.
    \item The left hand side trivializes combinatorially, in the sense that for each player $i$, letting $v_{i,1},\ldots,v_{i,N}$ be the challenges of player $i$ in $e_1,\ldots,e_N$ respectively, the formal expression $\prod\limits_{j=1}^{N} v_{i,j}^{s_j}$ reduces to the empty word by canceling any symbol $\sigma$ and any $\sigma^{-1}$ adjacent to it.
\end{enumerate}
\end{definition}

The main result of this section is the following lemma, which follows from~\cite[Theorem 5.10]{BWHK23}. We give a direct proof, generalizing the argument from~\cite[Theorem 2.1]{BeneWattsHelton2023} for $\mathbb{F}_2$, for the sake of completeness.
\begin{lemma}\label{lem:nocommute_orderedcontra}
    A generalized $k$-XOR $\Psi$ has no commuting embeddings if and only if it has an ordered contradiction.
\end{lemma}
\begin{proof}
    Suppose $\Psi$ has an ordered contradiction $(e_1,t_1),\ldots,(e_N,t_N)$,
    $(s_1,\ldots,s_N)$, and let $\{U_i(x)\}_{i\in[k],x\in\Sigma_i}$ be a commuting embedding with the vector $v\neq 0$.
    Writing $e_j = (x_1(j),\ldots,x_k(j))$, we have
    \[
    \prod\limits_{j=1}^{N}\prod\limits_{i=1}^{k}
    U_{i}(x_i(j))^{s_j} v
    =
    \prod\limits_{j=1}^{N}{\omega_p}^{s_jt_i}v\neq v.
    \]
    On the other hand, by commutation
    \[
    \prod\limits_{j=1}^{N}\prod\limits_{i=1}^{k}
    U_{i}(x_i(j))^{s_j} v
    =
    \prod\limits_{i=1}^{k}
    \prod\limits_{j=1}^{N}
    U_{i}(x_i(j))^{s_j} v
    =\prod\limits_{i=1}^{k} I v=v,
    \]
    so $\Psi$ has no commuting embedding.

    We now prove the other direction, and suppose there is no ordered contradiction. Define an expression $(w_1,\ldots,w_k,t)$ to be a tuple where $t\in \mathbb{F}_p$, and each $w_i$ is a reduced word over the alphabet $\Sigma_i$ (meaning, we allow inverses but cancel a symbol and its inverse if they are next to each other). We define the relation
    $\sim$ between expressions as follows: for each challenge $(x_1,\ldots,x_k,t)$ and $s\in\{-1,1\}$ we say
    \[
    (w_1,\ldots,w_k,z)\sim (w_1x_1^{-s},\ldots,w_kx_k^{-s},z-st),
    \]
    and then complete it to an equivalence relation. In this language, the fact that $\Psi$ has no ordered contradictions can be stated as asserting that the expressions $(\emptyset,\ldots,\emptyset,t)$ for $t=0,\ldots,p-1$ lie in different equivalence classes, and we denote the set of equivalency classes by $\mathcal{C}$.

    Consider the Hilbert space $V = \ell_2(\mathcal{C})$, and define the operators $U_i(x)\colon V\to V$ for each $i=1,\ldots,k$ and $x\in\Sigma_i$ by
    \[
     (U_i(x) v)_{[(w_1,\ldots,w_{i-1},x w_{i},w_{i+1},\ldots,w_k,t)]}
     =
     v_{[(w_1,\ldots,w_{i-1}, w_{i},w_{i+1},\ldots,w_k,t)]}
    \]
    We note that $U_{i}(x)$ acts bijectively on equivalence classes, and its inverse is the result of replacing $x$ by $x^{-1}$, so it is unitary. Let $v_a \in V$ be the elementary basis vector corresponding to the coordinate $(\emptyset,\ldots,\emptyset,a)$, and let $v = \sum\limits_{a\in\mathbb{F}_p}\omega_p^{a}v_a$. Since $(\emptyset,\ldots,\emptyset,a)$ all lie in different equivalence classes, the vectors $v_a$ are all linearly independent, and so $v\neq 0$. Lastly, for any challenge $(x_1,\ldots,x_k,t)$ we have
    \[
    \prod\limits_{i=1}^{k}U_i(x_i) v
    =\sum\limits_{a\in\mathbb{F}_p}\omega_p^{a}e_{[(x_1,\ldots,x_k,a)]}
    =
    \sum\limits_{a\in\mathbb{F}_p}\omega_p^{a}e_{[(\emptyset,\ldots,\emptyset,a-t)]}
    =\omega_p^{t} v,
    \]
    giving that they form a commuting embedding.
\end{proof}

\section{From Ordered Contradictions to Cauchy-Schwarz Proofs}\label{sec:gluing}
In this section we show how to convert an ordered contradiction as in Definition~\ref{def:ordered_contradiction} to a direct proof that the value of a game $\Psi$ is less than $1$ that only uses the Cauchy-Schwarz inequality.
\subsection{The Gluing Operation}
Let $\Psi$ be a generalized $k$-XOR game, and let $I\subseteq [k]$ be a non-empty set of players. We define the game $\Psi_{\text{glue}(I)}$, which has $2(k-\card{I})$ players, as follows:
\begin{enumerate}
    \item Sample a challenge $(x_1,\ldots,x_k,t)\sim \mu$ from $\Psi$.
    \item Sample a challenge $(x_1',\ldots,x_k',t')\sim \mu$ conditioned on $x_{i}'=x_i$ for all $i\in I$.
    \item For $i\not\in I$, send challenge $x_i$ to player $i$ and challenge $x_{i}'$ to player $i'$. Upon receiving their responses $a_{i}$ and $a_{i}'$, the verifier accepts if and only if
    \[
    \sum\limits_{i\not\in I} a_i+a_{i}' = t-t'.
    \]
\end{enumerate}
We also refer to this operation as gluing over $I$. We have the following lemma. It is easy to see that if $\mathsf{val}(\Psi)=1$, then $\mathsf{val}(\Psi_{\text{glue}(I)})=1$. The following claim shows that if $\mathsf{val}(\Psi_{\text{glue}(I)})<1$, then $\mathsf{val}(\Psi)<1$; in fact, importantly it shows that this follows by one simple application of the Cauchy-Schwarz inequality.
\begin{lemma}\label{lem:CS_once}
    Let $\Psi$ be a generalized $k$-XOR game, let $I\subseteq [k]$ be non-empty and denote by $\mu_{\Psi}$, $\mu_{\Psi_{\text{glue}(I)}}$ their distribution over challenges. Then for all $a\in\mathbb{F}_p$ and $f_1,\ldots,f_k$ with absolute value at most $1$,
    \[
    \card{\Expect{(x,t)\sim\mu_{\Psi}}{\prod\limits_{i=1}^{k}f_i(x_i)\omega_p^{-at}}}^2
    \leq
    \card{\Expect{(y,u)\sim\mu_{\Psi_{\text{glue}(I)}}}{\prod\limits_{i\not\in I}{f_{i}(y_{i})\overline{f_{i}(y_{i'})}\omega_p^{-au}}}}.
    \]
\end{lemma}
\begin{proof}
    Immediate by the Cauchy-Schwarz inequality.
\end{proof}
\subsection{Realizing Ordered Contradictions via Gluing}\label{sec:ordcont}
In this section we show that an iterative application of gluing operations can convert an ordered contradiction for a generalized $k$-XOR game $\Psi$ to a Cauchy-Schwarz based proof of $\mathsf{val}(\Psi)<1$. Towards this end, we introduce the useful notion of reading sequences.
Fix a generalized XOR game $\Psi$ that has ordered contradiction consisting of challenges $e_1,\ldots,e_N$ and signs $s_1,\ldots,s_N$. Reading the word player $i$ gets from left to right one symbol at a time, if the newly read symbol is the inverse of the right-most symbol we read (in the sense that they are the same symbol but the corresponding signs are opposite), we cancel both the right-most symbol and the newly read symbol. Denote by $h_i(j)$ the number of symbols remaining after reading the first $j$ symbols. For example, if the $j$th read letter was $\sigma$ with $s_j=1$ and the $(j+1)$-st letter was $\sigma$ with $s_{j+1}=-1$, then $h_{i}(j+1) = h_{i}(j)-1$. Define the sequence
\[
h_i=(h_i(0),h_i(1),\ldots,h_i(N)).
\]
\begin{definition}
Let $\Psi$ be a generalized XOR game that has an ordered contradiction of length $N$, and let $\mathcal{R}_N$ be the collection of all of these ordered contradictions. We define $H_N(\Psi)=\min_{\mathcal R\in\mathcal{R}_N}\max_i h_i$,
where the maximum is with respect to the lexicographic order. If $\Psi$ has no players, we define $H_N(\Psi)=(0,\ldots,0)$
\end{definition}

We now show that if $\Psi$ is a generalized XOR game that has an ordered contradiction of length $N$, then one can choose a set of players $I$ to glue over and decrease the parameter $H_N$. More precisely:
\begin{lemma}\label{lem:decreasing_gluing}
Suppose that $\Psi$ is a generalized $k$-XOR game for $k\geq 1$ that has an ordered contradiction of length $N$. Then there is $I\subseteq [k]$ such that $\Psi_{\text{glue}(I)}$ has an ordered contradiction of length $N$ and $H_N(\Psi_{\text{glue}(I)})< H_N(\Psi)$.
\end{lemma}
\begin{proof}
    Fix an ordered contradiction $e_1,\ldots,e_N$ with signs $s_1,\ldots,s_N$ attaining $H_N(\Psi)$. Denote the maximal sequence by $h$ and let $I$ be the set of players whose reading sequence is equal to $h$.

Consider the pairing of positions induced by $h$:
each position $j$ where $h$ increases is paired
with the position $j'>j$ where the corresponding
symbol is canceled. The pairs are non-crossing,
in the sense that they form a legal sequence
of opening and closing parentheses  `(' and `)'.
Since the pairing is determined by the sequence
of increases and decreases of $h$, it is the
same for all players in $I$.

Define $a\colon [N]\to[N]$ by mapping each
position to the opening position of its pair.
Thus, for every pair $(j,j')$ with $j<j'$,
we have $a(j)=a(j')=j$. We note that
$s_j=-s_{j'}$.

We now construct the ordered contradiction in $\Psi_{\text{glue}(I)}$. Towards this end, consider a pair of challenges $e_j$ and $e_{j'}$ from $\Psi$ where $a(j)=a(j')=j$. Take $e'_j$ to be the challenge in $\Psi_{\text{glue}(I)}$ resulting by picking $e_j$ in both items, and $e'_{j'}$ to be the challenge in $\Psi_{\text{glue}(I)}$ resulting by picking $e_{j'}$ in the first item and $e_{j}$ in the second item. Note that in the challenge $e_j'$, players $i$ and $i'$ receive the same question (which is the corresponding variable from $e_j$). In the challenge $e_{j'}'$, player $i$ gets as question the corresponding variable from $e_{j'}$, whereas player $i'$ gets as question the corresponding variable from $e_{j}$. We illustrate this construction, including its
iteration, in Appendix~\ref{app:gluing-example}.

We claim that $e'_1,\ldots,e'_N$, with the same signs $s_1,\ldots,s_N$, form an ordered contradiction.
Indeed, if $t_j$ is the right hand side of $e_j$, then the right hand side of $e'_j$ is $t_j-t_{a(j)}$, and we have
\[
\sum_{j=1}^N s_j(t_j-t_{a(j)})
=\sum_{j=1}^N s_jt_j
-\sum_{(a,b)}(s_a+s_b)t_a
=\sum_{j=1}^N s_jt_j\neq 0,
\]
where the second sum ranges over the pairs determined by $h$, and is $0$ because $s_a$, $s_b$ are different. As for the left hand sides, the word read by the
first copy of any player $i\notin I$ is unchanged,
so it cancels to the empty word. For the second
copy, at each pair of positions $(j,j')$ determined
by $h$, player $i'$ receives the same question
at both positions, with opposite signs.
Since the pairs are non-crossing, the word of
player $i'$ also cancels to the empty word,
by canceling pairs from the inside out.
Note that this need not be the greedy cancellation
pairing of player $i'$.
Thus, the new challenges still form an ordered
contradiction.

We next show that
$H_N(\Psi_{\mathrm{glue}(I)})<H_N(\Psi)$.
Fix $i\notin I$, and let $g$ be its reading
sequence, so that $g<h$.
The reading sequence of the first copy of
player $i$ is unchanged, hence is smaller than $h$.
Let $g'$ be the reading sequence of player $i'$.
We claim that $g'<h$.

Let $b$ be the first position where $g$ and $h$
differ. Since $g<h$, at position $b$ the sequence
$g$ decreases whereas $h$ increases.
We claim that players $i$ and $i'$ receive the
same questions at all positions up to and
including $b$.

Indeed, fix $b'<b$. If $h$ increases at $b'$,
then $a(b')=b'$, so both players receive the
same question. Otherwise, $h$ decreases at $b'$.
Since $g$ and $h$ agree before $b$, their
greedy cancellation pairings agree up to
that position. Thus, the questions player $i$
receives at $b'$ and $a(b')$ are equal,
so again both copies receive the same question.
Finally, $h$ increases at $b$, so $a(b)=b$,
and the questions agree there as well.

Consequently, $g'$ agrees with $g$ through
position $b$, inclusive. Since $g$ agrees with $h$ before
$b$, but $g(b)<h(b)$, we conclude that $g'<h$.
Thus, every player in the new ordered contradiction
has reading sequence strictly smaller than $h$,
completing the proof.
\end{proof}

Iterating Lemma~\ref{lem:decreasing_gluing} we get the following conclusion.
\begin{lemma}\label{lem:apply_CS_many_times}
    Suppose that $\Psi$ is a generalized $k$-XOR game that has an ordered contradiction of length $N$. Then there is $T = T(N)$ and a sequence of generalized XOR games $\Psi_0,\ldots,\Psi_T$ such that
    \begin{enumerate}
        \item $\Psi_0 = \Psi$.
        \item $\Psi_{i+1}$ is the result of $\Psi_{i}$ after gluing some set of players $I_i$.
        \item The game $\Psi_T$ has no players and at least two challenges have different right hand sides.
    \end{enumerate}
\end{lemma}
\begin{proof}
    Starting with $i=0$ and $\Psi_i = \Psi$, we apply Lemma~\ref{lem:decreasing_gluing} to find a set of players $I_i$ such that the game $\Psi_{i+1} = (\Psi_i)_{\text{glue}(I_i)}$ still has an ordered contradiction of length $N$ and $H_N(\Psi_{i+1})<H_N(\Psi_i)$. We increase $i$ by $1$ and then iterate. As for each $\Psi'$ that has an ordered contradiction of length $N$ we have that $H_N(\Psi')$ is a sequence of at most $N$ integers of size at most $N$,  we conclude that the process must terminate at some step $T$, in which $\Psi_T$ has an ordered contradiction of length $N$ and $H_N(\Psi_T) = \vec{0}$.
    Letting $s_1,\ldots,s_N$ be the signing in the contradiction and $t_1,\ldots,t_N$ be the right hand sides on $\Psi_T$, we conclude that $\sum_{i=1}^{N}s_i = 0$ (as all symbols cancel in $\Psi_{T-1}$ as it has at least one player), and $\sum\limits_{i=1}^{N}s_it_i\neq 0$, so not all of the $t_i$ can be equal.
 \end{proof}
\subsection{Tensorization}
Combining Lemma~\ref{lem:CS_once} and Lemma~\ref{lem:apply_CS_many_times} one can get a Cauchy-Schwarz based proof for $\mathsf{val}(\Psi)<1$. In the following lemma we give the tensorized version of this argument.
\begin{lemma}\label{lem:tensorize_CS}
    If $\Psi$ is a generalized $k$-XOR game that has an ordered contradiction of length $N$, then there is $\eps = \eps(N,\Psi)>0$ such that for all  $a_1,\ldots,a_n\in\mathbb{F}_p$ such that $a_i\neq 0$ for at least $d$ of the coordinates, and for all functions
    $f_i\colon \Sigma_i^{n}\to \mathbb{C}$ with absolute value at most $1$, we have that
    \[
    \card{\Expect{(x_1,\ldots,x_k,t)\sim\mu_{\Psi}^{\otimes n}}{\prod\limits_{i=1}^{k}f_i(x_i)
    \prod\limits_{i=1}^{n}\omega_p^{a_it_i}}}
    \leq (1-\eps)^d.
    \]
\end{lemma}
\begin{proof}
Let $M(d,\Psi,a_1,\ldots,a_n)$ be the maximum of the left hand side over all choices of $f_1,\ldots,f_k$.
Fix $\Psi_0,\ldots,\Psi_T$ from Lemma~\ref{lem:apply_CS_many_times}, and let $I_0,\ldots,I_{T-1}$ be the sets of players that get glued in each step. Note that by Lemma~\ref{lem:CS_once} there
are $1$-bounded functions $F_{i},F_{i'}\colon \Sigma_i^{n}\to \mathbb{C}$ for all $i\not\in I_0$, such that
\begin{align*}
\card{\Expect{(x_1,\ldots,x_k,t)\sim\mu_{\Psi}^{\otimes n}}{\prod\limits_{i=1}^{k}f_i(x_i)
\prod\limits_{i=1}^{n}\omega_p^{a_it_i}}}^2\leq
\card{\Expect{(x,x',t,t')\sim\mu_{\Psi_1}^{\otimes n}}{\prod\limits_{i\not\in I_0}F_i(x_i)F_{i'}(x_i')
\prod\limits_{i=1}^{n}\omega_p^{a_i(t_i-t_i')}}},
\end{align*}
which gives that $M(d,\Psi_0,a_1,\ldots,a_n)^2\leq M(d,\Psi_1,a_1,\ldots,a_n)$. Iterating, we get that
\[
M(d,\Psi_0,a_1,\ldots,a_n)^{2^T}\leq M(d,\Psi_T,a_1,\ldots,a_n).
\]
Note that as $\Psi_T$ has no players,
\[
M(d,\Psi_T,a_1,\ldots,a_n)
=\card{\Expect{(x,t)\sim \mu_{\Psi_T}^{\otimes n}}{\prod\limits_{i=1}^{n}\omega_p^{a_it_i}}}
=
\prod\limits_{i=1}^{n}\card{\Expect{(x,t)\sim \mu_{\Psi_T}}{\omega_p^{a_it}}}.
\]
As $\Psi_T$ has an ordered contradiction, for every $i$ such that $a_i\neq 0$ we get that
$\omega_p^{a_it}$ is not constant over the support of $\mu_{\Psi_T}$, so
$\card{\Expect{(x,t)\sim\mu_{\Psi_T}}{\omega_p^{a_it}}}<1$.
It follows that $M(d,\Psi_T,a_1,\ldots,a_n)\leq (1-\delta)^d$ for some $\delta = \delta(\Psi)>0$, so $M(d,\Psi,a_1,\ldots,a_n)\leq (1-\delta)^{d/2^{T}}\leq (1-\eps)^{d}$ where $\eps = \eps(\delta,T)>0$.
\end{proof}
\section{Proof of Theorem~\ref{thm:main}}
Fix a $k$-XOR game $\Psi$ as in Theorem~\ref{thm:main}, and let $\Psi'$ be the $2k$-XOR game from Lemma~\ref{lem:BLR_trick_strong}.
Then we conclude that $\Psi'$ has no commuting embeddings, so by Lemma~\ref{lem:nocommute_orderedcontra} it has an ordered contradiction.
\begin{claim}\label{claim:combine}
There exists $\eps = \eps(\Psi)>0$ such that
    $\mathsf{val}((\Psi')^{\otimes n})\leq (1-\eps)^n$ for all $n$.
\end{claim}
\begin{proof}
Let $A_1,\ldots,A_{2k}$ be any strategy for $(\Psi')^{\otimes n}$, we see that
\begin{align*}
&\Prob{(x_1,\ldots,x_{2k},t)\sim \mu_{\Psi'}^{\otimes n}}{A(x_1)_i+\ldots+A(x_{2k})_i=t_i~\forall i}\\
&\qquad
=\Expect{(x_1,\ldots,x_{2k},t)\sim \mu_{\Psi'}^{\otimes n}}{\prod\limits_{i=1}^{n}\frac{1}{p}\left(\sum\limits_{a_i\in\mathbb{F}_p}\omega_p^{a_i(A(x_1)_i+\ldots+A(x_{2k})_i-t_i)}\right)}\\
&\qquad
=\frac{1}{p^n}\sum\limits_{a_1,\ldots,a_n\in\mathbb{F}_p^n}
\Expect{(x_1,\ldots,x_{2k},t)\sim \mu_{\Psi'}^{\otimes n}}{f_{\vec{a},1}(x_1)\cdots f_{\vec{a},2k}(x_{2k}) \prod\limits_{i=1}^{n}\omega_p^{-a_it_i}},
\end{align*}
where $f_{\vec{a},j}(x_j) = \prod\limits_{i=1}^{n}\omega_p^{a_i A(x_j)_i}$. Using Lemma~\ref{lem:tensorize_CS} we get that the last sum is at most
\[
\frac{1}{p^n}\sum\limits_{a_1,\ldots,a_n\in\mathbb{F}_p^n}(1-\eps)^{\card{\sett{i}{a_i\neq 0}}}
\leq (1-\eps')^{n},
\]
where $\eps,\eps'>0$ depend only on $\Psi$.
\end{proof}

We next upper bound the value of the $n$-fold repeated version of $\Psi$ by the value of the $n$-fold repeated version of $\Psi'$.
\begin{claim}\label{claim:upperboundval_BLR}
    $\mathsf{val}(\Psi^{\otimes n})\leq \mathsf{val}((\Psi')^{\otimes n})$.
\end{claim}
\begin{proof}
    Letting $A_1,\ldots,A_k$ be a strategy for $\Psi^{\otimes n}$, we define the strategies
    \[
    B_{i}(x,\alpha) = -(\alpha_1 A_i(x)_1,\ldots,\alpha_n A_i(x)_n),
    \qquad
    B_{i'}(x,\alpha) = (\alpha_1 A_i(x)_1,\ldots,\alpha_n A_i(x)_n)
    \]
    for $\Psi'$. We note that for each $(x_1,\ldots,x_k)\sim \mu_{\Psi}^{\otimes n}$, if $A_1,\ldots,A_k$ succeed on challenges $x_1,\ldots,x_k$, then $B_1,B_{1'},\ldots,B_{k},B_{k'}$ succeed on $(x_1,\alpha(1)),(x_1,\alpha(1)+\gamma),\ldots,(x_{k},\alpha(k)), (x_k,\alpha(k)+\gamma)$ for all choices of $\alpha(1),\ldots,\alpha(k),\gamma\in\mathbb{F}_p^n$. The result now follows.
\end{proof}
Combining Claim~\ref{claim:combine} and Claim~\ref{claim:upperboundval_BLR} finishes the proof of Theorem~\ref{thm:main}.
\bibliographystyle{plain}
\bibliography{ref}
\newpage

\appendix
\section{An example: the Levi--Civita game}
\label{app:gluing-example}
\label{app:levi-civita-gluing}

The verifier in the Levi--Civita (permutation-sign) game samples $\pi\in S_3$,
sends $\pi(i)$ to player $i$, receives bits $a_i\in \mathbb{F}_2$, and accepts if $a_1+a_2+a_3=1_{\operatorname{sgn}(\pi)=1}\pmod{2}$; see also
\cite{braverman2026parity}. Its exponential parallel repetition bound
follows from the earlier Fourier-analytic and additive-combinatorial
approach of \cite[Theorem~2]{BhangaleEtAl2025}, and here we give an explicit Cauchy--Schwarz proof via the gluing
construction of Lemma~\ref{lem:decreasing_gluing} following the recipe in Section~\ref{sec:ordcont}.
Take the following six challenges and signs:
\[
\begin{array}{c|rrrrrr}
 j   & 1 & 2 & 3 & 4 & 5 & 6 \\ \hline
 e_j & 123 & 132 & 231 & 321 & 312 & 213 \\
 t_j & 0 & 1 & 0 & 1 & 0 & 1 \\
 s_j & +1 & -1 & +1 & -1 & +1 & -1
\end{array}
\]
Here $t_j\in\mathbb F_2$ are the targets, while the signs $s_j$
specify inverses. Inspecting the words each player receives, we see that they reduce to the empty word, and their reading sequences are
$h_1=0101210$, $h_2=0121010$, and $h_3=0123210$ respectively.
Since $\sum_{j=1}^6s_jt_j=1$ in $\mathbb F_2$, this is an ordered
contradiction.

The maximal reading sequence is $h=h_3$, so $I=\{3\}$.
Its pairing is $(1,6),(2,5),(3,4)$, giving
$(a(1),\ldots,a(6))=(1,2,3,3,2,1)$.
Following the proof, glue rows $1,2,3$ with themselves, and rows $4,5,6$
with rows $3,2,1$, respectively. Each pair agrees on player $3$'s
question. Deleting this player gives the following rows, with players
ordered as $1,2,1',2'$:
\[
\begin{array}{c|rrrrrr}
 j & 1 & 2 & 3 & 4 & 5 & 6 \\ \hline
 e'_j & 1212 & 1313 & 2323 & 3223 & 3113 & 2112 \\
 t_j-t_{a(j)} & 0 & 0 & 0 & 1 & 1 & 1
\end{array}
\]
For example, $e'_4$ glues $321$ with $231$, giving target $1-0=1$.

The signs $s_j$ and the signed target sum are unchanged.
The reading sequences, in player order $1,2,1',2'$, are now
\[
0101210,\qquad 0121010,\qquad 0101010,\qquad 0121210,
\]
all strictly smaller than $h=0123210$.
We remark that the words of players $1'$ and $2'$ admit the pairing of $h$, but it need not be their
greedy pairing: player $2'$ cancels positions $2,3$ immediately.
We also note that, as in the proof, both copies of player $2$ receive the same questions, namely $2,3,3$, until position $3$, which is the first position their reading sequences decrease while $h$ increases.

We continue running the process for a total of size times, and record below which player is being glued at each step, what is their reading sequence, and what is the pairing $a$.
The order of the players in a new row is their order in the previous row, where the primed players appear first.
\begin{center}
\renewcommand{\arraystretch}{1.08}
\begin{tabular}{c|c|c|c|c}
Gluings performed & Players & $h$ & $I$ & $(a(1),\ldots,a(6))$ \\
\hline
0 & 3  & $0123210$ & $\{3\}$     & $(1,2,3,3,2,1)$ \\
1 & 4  & $0121210$ & $\{4\}$     & $(1,2,2,4,4,1)$ \\
2 & 6  & $0121010$ & $\{2,5\}$   & $(1,2,2,1,5,5)$ \\
3 & 8  & $0101210$ & $\{1,3,4\}$ & $(1,1,3,4,4,3)$ \\
4 & 10 & $0101010$ & $[10]$      & $(1,1,3,3,5,5)$ \\
5 & 0  & $0000000$ & ---         & ---
\end{tabular}
\end{center}
After four gluings, rows $5$ and $6$ have the same question tuple
$(1,3,1,3,2,2,1,1,1,1)$, but targets $0$ and $1$, respectively.
The last gluing deletes all ten players, leaving targets $(0,0,0,0,0,1)$.
At every step the maximal reading sequence decreases and the signed
target sum remains $1$.
\end{document}